\documentclass[multi]{cambridge7A}
\usepackage[UKenglish]{babel}
\usepackage[longnamesfirst,sectionbib]{natbib}
\usepackage{chapterbib}
\usepackage{amsmath,amssymb,amsthm,amsfonts}
\usepackage{enumerate,url}
\usepackage{graphicx}

\setcitestyle{authoryear,round,semicolon}% alters natbib behaviour
\allowdisplaybreaks[1]

\theoremstyle{plain}
\newtheorem{theorem}{Theorem}[section]
\newtheorem{lemma}[theorem]{Lemma}

\theoremstyle{definition}
\newtheorem{definition}[theorem]{Definition}

\providecommand*\Index[1]{#1\index{#1}}

\providecommand*\Index[1]{#1\index{#1}}
\providecommand*\undex[1]{} % abandoned tag
\begin{document}
\alphafootnotes
\author[G. K. Ambler and B. W. Silverman]
{Graeme K. Ambler\footnotemark\
and Bernard W. Silverman\footnotemark}
\chapter[Perfect simulation using dominated coupling from the past]
{Perfect simulation using dominated coupling from the past with application
  to area-interaction point processes and wavelet thresholding}
\footnotetext[1]{Department of Medicine, Addenbrooke's Hospital, Hills
Road, Cambridge CB2~0QQ; graeme@ambler.me.uk}
\footnotetext[2]{Smith School of Enterprise and Environment,
Hayes House,
75 George Street,
Oxford
OX1 2BQ;   bernard.silverman@stats.ox.ac.uk}
\arabicfootnotes
\contributor{Graeme K. Ambler
  \affiliation{University of Cambridge}}
\contributor{Bernard W. Silverman
  \affiliation{University of Oxford}}

\renewcommand\thesection{\arabic{section}}
\numberwithin{equation}{section}
\renewcommand\theequation{\thesection.\arabic{equation}}
\numberwithin{figure}{section}
\renewcommand\thefigure{\thesection.\arabic{figure}}
\numberwithin{table}{section}
\renewcommand\thetable{\thesection.\arabic{table}}

\begin{abstract}
We consider perfect simulation algorithms for locally stable point processes 
based on dominated coupling from the past, and apply these methods in two different
contexts.   A new version of the algorithm is 
developed which is feasible for processes which are neither purely attractive 
nor purely repulsive.  Such processes include multiscale area-interaction 
processes, which are capable of modelling point patterns whose clustering 
structure varies across scales.    The other topic considered is nonparametric
regression
using wavelets, where we use a suitable area-interaction process on the discrete space of
indices of wavelet coefficients to model the notion that if one wavelet coefficient is non-zero
then it is more likely that neighbouring coefficients will be also.   A method based on 
perfect simulation within this model shows promising results compared to the standard methods which threshold
coefficients independently.    
\end{abstract}

\subparagraph{Keywords}coupling from the past (CFTP), dominated CFTP, exact simulation, local stability, Markov chain Monte Carlo, perfect simulation, Papangelou conditional intensity, spatial birth-and-death process
\subparagraph{AMS subject classification (MSC2010)}62M30, 60G55, 60K35

\section{Introduction}
\index{Markov, A. A.!Markov chain Monte Carlo
(MCMC)}\index{simulation!exact/perfect simulation|(}Markov chain Monte Carlo (MCMC) is now one of the standard approaches of 
computational Bayesian inference\index{Bayes, T.!Bayesian inference}.  
A standard issue when using MCMC is the need to ensure that the Markov chain we are using for simulation has reached equilibrium.
For certain classes of problem, this problem was solved 
by the introduction of \index{coupling!coupling from the past (CFTP)}coupling from the past (CFTP)
\citep{pro-wil:exa,pro-wil:unk-mar}.  More recently, methods based on 
CFTP have been developed for perfect simulation of spatial point
process\index{point process!spatial point process} models 
(see for example \cite{ken:boo,ken:are,hag-lie-mol:exa-spa,ken-mol:per}).  

Exact CFTP methods are therefore attractive, as one does not need to 
check convergence rigorously or worry about burn-in, or use complicated methods to find 
appropriate standard errors for Monte Carlo estimates based on correlated samples.  
Independent and identically distributed samples are now available, so estimation 
reduces to the simplest case.  Unfortunately, this simplicity comes at a price.  
These methods are notorious for taking a long time to return just one exact 
sample and are often difficult to code, leading many to give up and return to 
nonexact methods.  In response to these issues, in the first part of this paper we present a 
dominated CFTP algorithm\index{coupling!dominated CFTP algorithm} for the
simulation of locally stable point processes\index{point process!locally stable point process} 
which potentially requires far fewer evaluations per iteration than the existing 
method in the literature \citep{ken-mol:per}.

The paper then goes on to discuss applications of this CFTP algorithm, in two different contexts, 
the modelling of point patterns and nonparametric regression\index{regression!nonparametric regression} by wavelet thresholding\index{wavelet!thresholding}.  
In particular it will be seen that these two problem areas are much more closely related than might be imagined, 
because of the way that the non-zero coefficients in a wavelet expansion may be modelled as an appropriate point process.

The structure of the paper is as follows.  In Section~\ref{sec:simulation} 
we discuss perfect simulation, beginning with ordinary 
coupling from the past (CFTP) and moving on to dominated CFTP for spatial point 
processes.  We then introduce and justify our perfect simulation algorithm.  In 
Section \ref{aiprocess} we first review the standard \index{point process!area-interaction process}area-interaction process.  We 
then introduce our \index{point process!multiscale area-interaction process}multiscale process, describe how to use our new perfect 
simulation algorithm\index{simulation!exact/perfect simulation} to simulate from it, and discuss a method for inferring 
the parameter values from data, and present an application to the \index{data!Redwood seedlings}Redwood seedlings data.
In Section~\ref{sec:wavelets} we turn attention to the wavelet regression
problem\index{wavelet!regression}.
Bayesian approaches are reviewed, and a model introduced which incorporates an 
area-interaction process on the discrete space of indices of wavelet coefficients. 
In Section~\ref{sec:wavperfsim} the application of our perfect simulation algorithm in
this context is developed.   The need appropriately to modify the approach  to increase
its computational feasibility is addressed, and a simulation study investigating 
its performance on standard test examples is carried out.  Sections~\ref{aiprocess} and \ref{sec:wavperfsim} both
conclude with some suggestions for future work.

\section{Perfect simulation}\label{sec:simulation}
\subsection{Coupling from the past}\label{cftp}

%%%%  New material begins

In this section, we offer a brief intuitive introduction to the 
principle behind CFTP. 
For more formal descriptions and details, see, for example,  
\citet{pro-wil:exa}\index{Propp, J. G.}\index{Wilson, D. B.}, \citet[Chapter~32]{mac:inf} and \citet{con:per}.

Suppose we wanted to
sample from the stationary distribution of an irreducible aperiodic 
Markov chain $\{Z_t\}$
on some (finite) state space $X$ with states 1, \ldots, $n$.  Intuitively,
if it were possible to go back an infinite amount in time and start the 
chain
running, the chain
would be in its stationary distribution when one
returned to the present (i.e. $Z_0\sim\pi$, where $\pi$ is the stationary
distribution of the chain).

Now, suppose we were to set not one, but $n$ chains
$\{Z^{(1)}_t\}$, \ldots, $\{Z^{(n)}_t\}$ running at a fixed time $-M$ in the
past, where $Z^{(i)}_{-M}=i$ for each chain $\{Z^{(i)}_t\}$.  Now let all
the chains be coupled so that if $Z^{(i)}_{s}=Z^{(j)}_{s}$ at any time $s$
then $Z^{(i)}_{t}=Z^{(j)}_{t}$ $\forall t\geq s$.  Then if
all the chains ended up in the same state $j$ at time zero
(i.e. $Z^{(i)}_0=j$ $\forall i\in X$), we would know that
whichever state the chain passing from time minus infinity to zero was in
at time $-M$, the chain would end up in state $j$ at time zero.  Thus 
the state at time zero
is a sample from the stationary distribution provided $M$ is large 
enough for coalescence to have been achieved for the realisations being 
considered. 

%%%%%%%%%%%%%%%%  NEW MATERIAL ENDS:  Deleted material follows

%In this section, we set out briefly the principle behind CFTP.  
%For a more detailed description and discussion, see, for example, \citet[Chapter~32]{mac:inf} %and \citet{con:per}.

%Suppose that it is desirable to 
%sample from the stationary distribution of an ergodic Markov chain $\{Z_t\}$ 
%on some (finite) state space $X$ with states 1, \ldots, $n$.  It is clear that 
%if it were possible to go back an infinite amount in time, start the chain 
%running (in state $Z_{-\infty}$) and then return to the present, the chain 
%would (with probability 1) be in its stationary distribution when one 
%returned to the present (i.e. $Z_0\sim\pi$, where $\pi$ is the stationary 
%distribution of the chain).

%Now, suppose we were to set not one, but $n$ chains 
%$\{Z^{(1)}_t\}, \ldots, \{Z^{(n)}_t\}$ running at a fixed time $-M$ in the 
%past, where $Z^{(i)}_{-M}=i$ for each chain $\{Z^{(i)}_t\}$.  Now let all 
%the chains be coupled so that if $Z^{(i)}_{s}=Z^{(j)}_{s}$ at any time $s$ 
%then $Z^{(i)}_{t}=Z^{(j)}_{t}$ $\forall t\geq s$.  Then if 
%all the chains ended up in the same state $j$ at time zero 
%(i.e. $Z^{(i)}_0=j$ $\forall i\in X$), we would know that 
%whichever state the chain passing from time minus infinity to zero was in 
%at time $-M$, the chain would end up in state $j$ at time zero.  Thus $j$ 
%must be a sample from the stationary distribution of the Markov chain in 
%question.

When performing CFTP, a useful property of the coupling chosen is that it be
\emph{stochastically monotone}\index{stochastically monotone} as in the following definition.

\begin{definition}\label{sto-mono}
Let $\{Z^{(i)}_t\}$ and $\{Z^{(j)}_t\}$ be two Markov chains obeying the 
same transition kernel.  Then a coupling of these Markov chains is 
stochastically monotone with respect to a partial ordering $\preceq$ if 
whenever $Z^{(i)}_t\preceq Z^{(j)}_t$, then $Z^{(i)}_{t+k}\preceq Z^{(j)}_{t+k}$ 
for all positive $k$.
\end{definition}

Whenever the coupling used is stochastically monotone and there are maximal 
and minimal elements with respect to $\preceq$ then we need only simulate 
chains which start in the top and bottom states, since chains starting in all 
other states are sandwiched by these two.  This is an important ingredient 
of the dominated coupling from the past\index{coupling!dominated coupling from the past|(} algorithm introduced in the next 
section.

Although 
attempts have been made to generalise CFTP to continuous state 
spaces (notably \cite{mur-gre:con} and \cite{gre-mur:bay}, 
as well as \cite{ken-mol:per}, discussed in Section~\ref{domcftp}), 
there is still much work to be done before exact sampling becomes 
universally, or even generally applicable.  For example, there are no truly 
general methods for processes in high, or even moderate, dimensions.

\subsection{Dominated coupling from the past}
\label{domcftp}

Dominated coupling from the past was introduced as an extension of coupling 
from the past which allowed the simulation of the area-interact\-ion process 
\citep{ken:are}\index{point process!area-interaction process|(}, though it was soon extended to other types of point processes 
and more general spaces \citep{ken-mol:per}\index{Kendall,
W. S.|(}\index{Moller, J.@M{\o}ller, J.|(}.  We give the formulation for 
locally stable point processes\index{point process!locally stable point process}.

%Suppose that we wish to obtain a sample of a spatial point process\index{point
%process!spatial point process} with density 
%$f$ with respect to the unit rate \index{Poisson, S. D.!Poisson process|(}Poisson process,
%whose \index{Papangelou, F.!Papangelou conditional intensity}Papangelou 
%conditional intensity, $\lambda_f(u;X)$, is uniformly bounded above by some 
%constant $\lambda$:
%\[
%	 \lambda_f(u;X)=\frac{f(X\cup\{u\})}{f(X)}\leq\lambda.
%\]
%The uniform bound on the Papangelou conditional intensity is required in order 
%for the point process to be locally stable, so this is not imposing any 
%additional constraint.
%%%%%%%%%%%%%%%%%%%DELETED above paragraph:  NEW MATERIAL BEGINS

Let $x$ be a spatial point 
pattern 
in some bounded subset
$S\subset\mathbb{R}^n$, and $u$ a single point $u\in S$.  Suppose that $x$
is a realisation of a spatial point process\index{point process!spatial point process} $X$
with density $f$ with respect to the unit rate 
\index{Poisson, S. D.!Poisson process|(}Poisson process.
The
\index{Papangelou, F.!Papangelou conditional intensity}\emph{Papangelou conditional intensity} $\lambda_f$ is defined by
\[
\lambda_f(u; x)  = \frac{f(x\cup\{u\})}{f(x)};
\]
see, for example, \cite{pap:con}\index{Papangelou, F.} and
\cite{bad-etal:res}\index{Baddeley, A. J.}\index{Turner, R.}\index{Hazelton, M.}.  If the process $X$ is locally stable, then there exists a constant $\lambda$ such that $\lambda_f(u;x) \leq \lambda$ for all finite point configurations $x\subset S$ and all points $u\in S \setminus x$.

%%%%%%%%%%%%%NEW MATERIAL ENDS

The algorithm given in \cite{ken-mol:per} is then as follows.

\begin{enumerate}
\item Obtain a sample of the Poisson process with rate $\lambda$.
\item Evolve a \index{Markov, A. A.!Markov process}Markov process $D(T)$ \emph{backwards} until some fixed time $-T$, 
using a \Index{birth-and-death process} with death rate equal to 1 and birth rate equal 
to $\lambda$.  The configuration generated in step 1 is used as the initial state.
\item Mark all of the points in the process with U[0,1] marks.  We refer to the 
mark of point $x$ as $P(x)$.
\item \label{init} Recursively define upper and lower processes, $U$ and $L$ as 
follows.  The initial configurations at time $-T$ for the processes are 
\begin{eqnarray*}
	U_{-T}(-T) & = & \left\{x\; : x\in D(-T)\right\};\\
	L_{-T}(-T) & = & \left\{\boldsymbol{0}\right\}.
\end{eqnarray*}
\item \label{evolve} Evolve the processes \emph{forwards} in time to $t=0$ in the 
following way.

Suppose that the processes have been generated up a given time, $u$, and suppose that 
the next birth or death to occur after that time happens at time $t_i$.  If 
a \textbf{birth} happens next then we accept the birth of 
the point $x$ in $U_{-T}$ or $L_{-T}$ if the point's mark, $P(x)$, is less than

\begin{equation} \label{birthrate}
\begin{split}
&\min\left\{\frac{\lambda_f(x;X)}{\lambda}:
		L_{-T}(t_i)\subseteq X\subseteq U_{-T}(t_i)\right\}\ \text{or}\\
&\max\left\{\frac{\lambda_f(x;X)}{\lambda}:
		L_{-T}(t_i)\subseteq X\subseteq U_{-T}(t_i)\right\}
\end{split}
\end{equation}
respectively, where $x$ is the point to be born.

If, however, a \textbf{death} happens next then if the event is present 
in either of our processes we remove the dying event, setting
$U_{-T}(t_i) = U_{-T}(u)\setminus \{x\}$ and 
$L_{-T}(t_i) = L_{-T}(u)\setminus \{x\}$.
\item Define $U_{-T}(u+\varepsilon) = U_{-T}(u)$ and 
$L_{-T}(u+\varepsilon) = L_{-T}(u)$ for $u<u+\varepsilon<t_i$.
\item If $U_{-T}$ and $L_{-T}$ are identical at time zero 
(i.e. if $U_{-T}(0) = L_{-T}(0)$), then we have the required sample 
from the \index{point process!area-interaction process|)}area-interaction process with rate parameter $\lambda$ and 
attraction parameter $\gamma$.  If not, go to step 2 and repeat, extending 
the underlying \index{Poisson, S. D.!Poisson process|)}Poisson process back to $-(T+S)$ and generating 
additional $U[0,1]$ marks (keeping the ones already generated).
\end{enumerate}

This algorithm involves calculation of $\lambda(u;X)$ for each 
configuration that is both a subset of $U(T)$ and a superset of $L(T)$.  
Since calculation of $\lambda(u;X)$ is typically expensive, this calculation 
may be very costly.  The method proposed in Section~\ref{perfect_alg} uses 
an alternative version of step~\ref{evolve} which requires us only to 
calculate $\lambda(u;X)$ for upper and lower processes.

The more general form given in \cite{ken-mol:per}\index{Kendall, W. S.!Kendall--M{\o}ller algorithm|(} may be obtained from the 
above algorithm by replacing the evolving Poisson process $D(T)$ with a 
general dominating process on a partially ordered space $(\Omega,\preceq)$ 
with a unique minimal element $\boldsymbol{0}$.  The partial ordering in the 
above algorithm is that induced by the subset relation $\subseteq$.  
Step~\ref{evolve} is replaced by any step which preserves the crucial \emph{funnelling property}
\begin{equation}
\label{funnel}	L_{-T}(u)\preceq L_{-(T+S)}(u)\preceq U_{-(T+S)}(u)\preceq U_{-T}(u)
\end{equation}
for all $u<0$ and $T$, $S>0$
and the \emph{sandwiching relations}
\begin{eqnarray}
\label{sand1} &	L_{-T}(u)\preceq X_{-T}(u)\preceq U_{-T}(u)\preceq D(u) &
\text{and}\\
\label{sand2} &	L_{-T}(t)=U_{-T}(t)\;\text{ if }\;L_{-T}(s)=U_{-T}(s) &
\end{eqnarray}
for $s\leq t\leq0$.
In equation (\ref{sand1}), $X_{-T}(u)$ is the 
Markov chain or process from whose stationary distribution we wish to
sample.\index{coupling!dominated coupling from the past|)}

\subsection{A perfect simulation
algorithm}\label{perfect_alg}\index{algorithm!perfect simulation algorithm|(}

Suppose that we wish to sample from a locally stable point process\index{point
process!locally stable point process|(} with density
\begin{equation}\label{gen-proc}
	p(X) = \alpha\prod_{i=1}^m f_i(X),
\end{equation}
where $\alpha\in(0,\infty)$ and $f_i:\mathfrak{R}^f\to\mathbb{R}$ are 
positive-valued functions which are monotonic with respect to the 
partial ordering $\preceq$ induced by the subset relation\footnote{That is, 
configurations $x$ and $y$ satisfy $x\preceq y$ if $x\subseteq y$.} and 
have uniformly bounded Papangelou conditional intensity\index{Papangelou, F.!Papangelou conditional intensity}:
\begin{equation}
\label{eq:papacond}
	\lambda_{f_i}(u;\mathbf{x})
		=\frac{f_i(\mathbf{x}\cup\{u\})}{f_i(\mathbf{x})}
		\leq K.
\end{equation}

When the conditional intensity (\ref{eq:papacond}) can be
expressed in this way, as the product of monotonic
interactions, then we shall demonstrate that the crucial step of the
Kendall--M{\o}ller algorithm may be re-written
in a form which is computationally much more
efficient, essentially by dealing with each
factor separately. 

Clearly
\begin{equation}\label{gen-dom}
	\lambda_p(u;\mathbf{x}) \leq \lambda = 
		\prod_{i=1}^m\max_{X,\{x\}}\lambda_{f_i}(x;X)
\end{equation}
for all $u$ and $\mathbf{x}$, and $\lambda$ is finite.  Thus we may use 
the algorithm in Section~\ref{domcftp} to simulate from this process 
using a Poisson process with rate $\lambda$ as the dominating process.

However, as previously mentioned, calculation of $\lambda_p(u;\mathbf{x})$ 
is typically expensive\index{computational complexity|(}, increasing at least linearly in $n(\mathbf{x})$.  
Thus to calculate the expressions in (\ref{birthrate}), we must in 
general perform $2^{n(U_{-T}(t_i))-n(L_{-T}(t_i))}$ of these 
calculations, making the algorithm non-polynomial.  In practice it is 
clearly not feasible to use this algorithm in all but the most trivial of 
cases, so we must look for some way to reduce the computational burden in 
step~\ref{evolve} of the algorithm.   

This can be done by replacing step~\ref{evolve} with the following alternative\index{computational complexity|)}.

\begin{enumerate}
\item[\ref{evolve}$\mbox{}^\prime$] Evolve the processes \emph{forwards} in 
time to $t=0$ in the following way.

Suppose that the processes have been generated up a given time, $u$, and suppose 
that the next birth or death to occur after that time happens at time $t_i$.  If 
a \textbf{birth} happens next then we accept the birth of 
the point $x$ in $U_{-T}$ or $L_{-T}$ if the point's mark, $P(x)$, is less than
\begin{eqnarray}
\label{gen-up}	& \prod_{i=1}^m \left[\max\left\{ \lambda_{f_i}(u;U(T)) , 
			\lambda_{f_i}(u;L(T))\right\}
		\left/ \lambda\right.\right] 
								&\text{ or}\\
\label{gen-low} & \prod_{i=1}^m \left[\min\left\{ \lambda_{f_i}(u;U(T)) , 
			\lambda_{f_i}(u;L(T))\right\}
		\left/ \lambda\right.\right] &
\end{eqnarray}
respectively, where $x$ is the point to be born.

If, however, a \textbf{death} happens next then if the event is present 
in either of our processes we remove the dying event, setting
$U_{-T}(t_i) = U_{-T}(u)\setminus \{x\}$ and 
$L_{-T}(t_i) = L_{-T}(u)\setminus \{x\}$.
\end{enumerate}

\begin{lemma}\label{cftp_lemma}
Step\/ $\ref{evolve}'$ obeys properties\/ $(\ref{funnel})$, $(\ref{sand1})$ and $(\ref{sand2})$,
and is thus a valid dominated coupling-from-the-past 
algorithm.
\end{lemma}
\begin{proof} 
Property~(\ref{sand1}) follows by noting 
that
\[
(\ref{gen-low})\leq\lambda_p(u;X)\leq(\ref{gen-up})\leq1.
\]
Property~(\ref{sand2}) is trivial.  Property~(\ref{funnel}) follows 
from the monotonicity of the $f_i$.
\end{proof}

\begin{theorem}
Suppose that we wish to simulate from a locally stable point
process\index{point process!locally stable point process|)} whose 
density $p(X)$ with respect to the unit-rate \index{Poisson, S. D.!Poisson process}Poisson process is representable 
in form\/ $(\ref{gen-proc})$.  Then by replacing Step\/ $\ref{evolve}$ by 
Step\/ $\ref{evolve}'$ it is possible to bound the necessary number of 
calculations of $\lambda_p(u;X)$ per iteration in the dominated 
coupling-from-the-past algorithm independently of $n(X)$.
\end{theorem}
\begin{proof} Step~\ref{evolve}$\mbox{}^\prime$ clearly involves only a 
constant number of calculations, so by
Lemma~\ref{cftp_lemma} above and Theorem~2.1 of \cite{ken-mol:per}, the 
result holds.
\end{proof}

In the case where it is possible to write $p(X)$ in form~(\ref{gen-proc}) 
with $m=1$, Step~\ref{evolve}$\mbox{}^\prime$ is identical to Step~\ref{evolve}.  This is 
the case for models which are either purely attractive or purely
repulsive\index{point process!clustering|(}, 
such as the standard \index{point process!area-interaction process|(}area-interaction process discussed in 
Section~\ref{std_aip}.  It is not the case for the \index{point process!multiscale area-interaction process}multiscale process
discussed in Section~\ref{multi}, or the model for \index{wavelet!thresholding}wavelet coefficients discussed
in Section~\ref{model}.  

The proof of 
Theorem~2.1 in \cite{ken-mol:per} does not require that the initial 
configuration of $L_{-T}$ be the minimal element $\mathbf{0}$, only that 
it be constructed in such a way that properties~(\ref{funnel}), (\ref{sand1}) and (\ref{sand2}) are satisfied.  
Thus we may refine our method further by 
modifying step~\ref{init} so that the initial configuration of $L_{-T}$ is 
given by
\begin{equation}\label{gen-thin}
	L_{-T}(-T) = \left\{x\in D(-T):P(x)\leq\prod_{i=1}^m\left[\min_{X,\{x\}}
		\lambda_{f_i}(x;X)\left/\lambda\right.\right]\right\},
\end{equation}
which clearly satisfies the necessary requirements\index{algorithm!perfect
simulation algorithm|)}\index{Kendall, W. S.!Kendall--M{\o}ller
algorithm|)}\index{Kendall, W. S.|)}\index{Moller, J.@M{\o}ller, J.|)}.

\section{Area-interaction processes}\label{aiprocess}

\subsection{Standard area-interaction process}\label{std_aip}

%\subsection{The Area-interaction point process}\label{are-int}
There are several classes of model for 
stochastic point processes, for example simple \index{Poisson, S. D.!Poisson
process}Poisson processes, cluster 
processes such as \index{Cox, D. R.!Cox process}Cox processes, and processes defined as the stationary 
distribution of \index{Markov, A. A.!Markov point process}Markov point processes, such as Strauss processes 
\citep{str:a-mod} and area-interaction processes \citep{bad-lie:are-int}.  
The area-interaction point process is capable of producing both moderately clustered and moderately ordered
patterns depending on the value of its clustering parameter.  It was
introduced primarily to fill a gap left by the \index{Strauss, D. J.!Strauss
point process}Strauss point process
\citep{str:a-mod}, which can produce only ordered point patterns
\citep{kel-rip:a-not}.

The general definition of the area-interaction process depends on a specification of
the \emph{neighbourhood} of any point in the space $\chi$ on which the process
is defined.  Given any $x\in\chi$ we denote by $B(x)$ the neighbourhood of the
point $x$.  Given a set $X\subseteq\chi$, the neighbourhood $U(X)$ of $X$ is
defined as $\bigcup_{x\in X}B(x)$.  The general area-interaction process
is then defined by \cite{bad-lie:are-int}\index{Baddeley, A. J.|(}\index{Lieshout, M. N. M. van|(} as follows.

Let $\chi$ be some locally compact complete metric space and
$\mathfrak{R}^f$ be the space of all possible configurations of points in
$\chi$.  Suppose that $m$ is a finite Borel regular measure on $\chi$ and
$B:\chi\to \mathcal{K}$ be a \Index{myopically continuous function} 
\citep{mat:ran}, where $\mathcal{K}$ is the class of all compact subsets 
of $\chi$.  Then the probability density of the general area-interaction 
process is given by
\begin{equation}\label{gaip}
	p(X) = \alpha\lambda^{N(X)}\gamma^{-m\{U(X)\}}
\end{equation}
with respect to the unit rate Poisson process\index{Poisson, S. D.!Poisson process}, where $N(X)$ is the number of
points in configuration $X=\{x_1,\ldots,x_{N(X)}\}\in\mathfrak{R}^f$, $\alpha$
is a normalising constant and $U(X) = \bigcup_{i=1}^{N(X)}B(x_i)$ as above.

In the spatial point-process\index{point process!spatial point process} case, for some fixed compact set $G$ in $\mathbb{R}^d$, 
the neighbourhood $B(x)$ of each point $x$ is defined to be $x \oplus G$.
Here $\oplus$ is the \index{Minkowski, H.!Minkowski addition}Minkowski addition operator, defined by $A \oplus B = \{ a+b :  a \in A, b \in B \}$
for sets $A$ and $B$.   So the resulting 
area-interaction process has density 
\begin{equation}\label{aip}
	p(X) = \alpha\lambda^{N(X)}\gamma^{-m(X\oplus G)}
\end{equation}
with respect to the unit-rate Poisson process, where $\alpha$ is a 
normalising constant, $\lambda>0$ is the \emph{rate} parameter, $N(X)$ is 
the number of points in the configuration $X$, $\gamma>0$ is the 
\emph{clustering} parameter.  Here $0<\gamma<1$ is the \emph{repulsive} case, while $\gamma>1$ is the 
\emph{attractive} case.  The case $\gamma=1$ reduces to the homogeneous 
Poisson process with rate $\lambda$.  Figure \ref{fig:AIP} gives an example of the construction when $G$ is a 
disc.

\begin{figure}
	\begin{center}
		\resizebox{0.6\textwidth}{!}{\hspace{-6em}\includegraphics{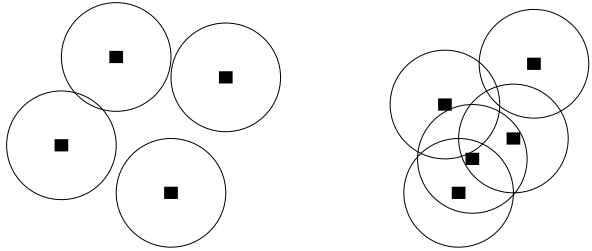}}
	\end{center}
		\caption[An example of some events together with circular `grains' G.]{An example of some events together with circular `grains' $G$.  The events in the above diagram would be the actual members of the process.  The circles around them are to show what the set $X\oplus G$ would look like.  If $\gamma$ were large, the point configuration on the right would be favoured, whereas if $\gamma$ were small, the configuration on the the left would be favoured.\label{fig:AIP}}
%\vspace{1ex}
\end{figure}

\subsection{A multiscale area-interaction process}\label{multi}\index{point process!multiscale area-interaction process|(}

The area-interaction process is a flexible model yielding a good   
range of models, from regular through total spatial randomness to 
clustered.  Unfortunately it does not allow for models whose behaviour 
changes at different resolutions, for example repulsion at small 
distances and attraction at large distances.  Some examples which display
this sort of behaviour are the distribution of trees 
on a hillside, or the distribution of \Index{zebra} in a patch of savannah.
A physical example of large scale attraction and small scale repulsion is 
the interaction between the \Index{strong nuclear force} and the
\Index{electro-magnetic force} between two oppositely charged particles.  The physical 
laws governing this behaviour are different from those governing the 
behaviour of the area-interaction class of models, though they may be 
sufficiently similar so as to provide a useful approximation.

We propose the following model to capture these types of behaviour.
%In an attempt to develop a model with these characteristics.

\begin{definition}
The \emph{multiscale area-interaction process} has 
density
\begin{equation}\label{att-rep}
	p(X) = \alpha\lambda^{N(X)}\gamma_1^{-m(X\oplus G_1)}
					\gamma_2^{-m(X\oplus G_2)},
\end{equation}
where $\alpha$, $\lambda$ and $N(X)$, are as in equation \eqref{aip}; 
$\gamma_1\in[1,\infty)$ and $\gamma_2\in(0,1]$; and
$G_1$ and $G_2$ are balls of radius $r_1$ and $r_2$ respectively.
\end{definition}

The process is clearly Markov of range $\max\{r_1,r_2\}$.  If 
$G_1\supset G_2$, we will have small-scale repulsion and large-scale 
attraction.  If $G_1\subset G_2$, we will have small-scale attraction and 
large-scale repulsion.

\begin{theorem}\label{thm:att-rep}
The density \eqref{att-rep} is both measurable and integrable.
\end{theorem}

This is a straightforward extension of the proof of \cite{bad-lie:are-int}
for the standard area-interaction process\index{Baddeley, A. J.|)}\index{Lieshout, M. N. M. van|)};
for details, see the Appendix of \cite{amb-sil:per}.

\subsection{Perfect simulation of the multiscale process}\label{perfectmulti}

Perfect simulation of the multiscale process (\ref{att-rep}) is possible 
using the method introduced in Section~\ref{perfect_alg}.  Since 
(\ref{att-rep}) is already written as a product of three monotonic 
functions with uniformly bounded Papangelou conditional
intensities\index{Papangelou, F.!Papangelou conditional intensity}, we need 
only substitute into equations (\ref{gen-dom}--\ref{gen-thin}) as follows.

Substituting 
into equation (\ref{gen-dom}), we find that the rate of a suitable 
dominating process is
\[
	\lambda\gamma_2^{-m(G_2)}.
\]
The initial configurations of the upper and lower processes $U$ and $L$ are 
then found by simulating this process, thinning with a probability of 
\[
	\gamma_1^{-m(G_1)}\gamma_2^{m(G_2)}
\]
for $L$.

As $U$ and $L$ evolve towards time $0$, we accept points $x$ in $U$ with 
probability
\begin{equation}\label{multiUaccept}
		\gamma_1^{-m((x\oplus G_1)\setminus U_{-T}(u)\oplus G_1)}
		\gamma_2^{m(G_2)-m((x\oplus G_2)\setminus L_{-T}(u)\oplus G_2)}
\end{equation}
and accept events in $L$ whenever 
\begin{equation}\label{multiLaccept}
		P(x)\leq\gamma_1^{-m((x\oplus G_1)\setminus L_{-T}(u)\oplus G_1)}
		\gamma_2^{m(G_2)-m((x\oplus G_2)\setminus U_{-T}(u)\oplus G_2)}.
\end{equation}

Figure \ref{fig:AIP2} 
gives examples of the construction $(x\oplus G)\setminus Y_{-T}(u)\oplus G$\index{point process!multiscale area-interaction process|)}.
\begin{figure}[t]
	\begin{center}
		\resizebox{0.9\textwidth}{!}{\includegraphics{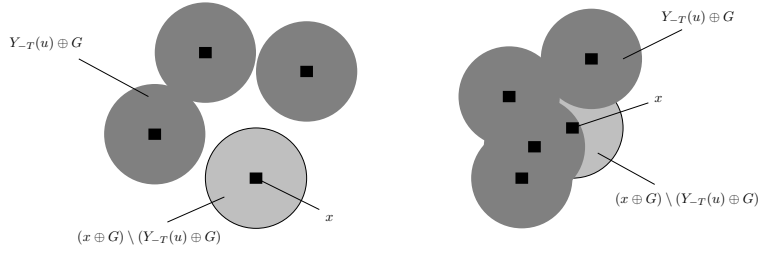}}
	\end{center}
		\caption[Another look at Figure \ref{fig:AIP}.]{Another look at Figure \ref{fig:AIP} with some shading added to show the process of simulation.  Dark shading shows $Y_{-T}(u)\oplus G$ where $Y_{-T}(u)$ is the state of either $U$ or $L$ immediately before we add the new event and $G$ could be either $G_1$ or $G_2$.  Light shading shows the amount added if we accept the new event.  In the configuration on the left, $x\oplus G=(x\oplus G)\setminus(Y_{-T}(u)\oplus G)$, so that the attractive term in (\ref{multiUaccept}) or (\ref{multiLaccept}) will be very small, whereas the repulsive term will be large.  In the configuration on the right we are adding very little area to $(Y_{-T}(u)\oplus G)$ by adding the event, so the attractive term will be larger and the repulsive term will be smaller.\label{fig:AIP2}}
%\vspace{1ex}
\end{figure}

\subsection{Redwood seedlings data}\label{appl}\index{data!Redwood seedlings|(}

\begin{figure}
 \begin{center}
  \resizebox{.95\textwidth}{!}{\hspace{-5em}\resizebox{.56\textwidth}{!}{\includegraphics{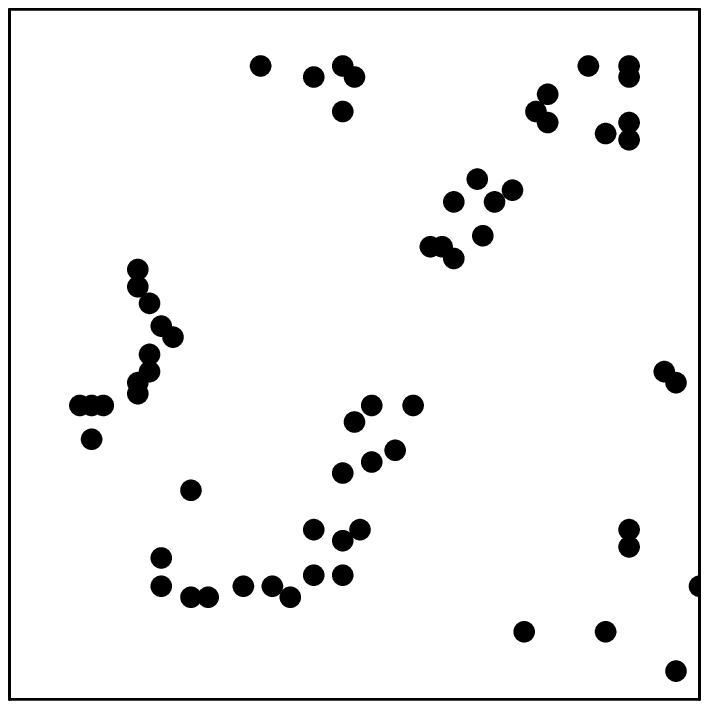}}\hspace{2em}\resizebox{.56\textwidth}{!}{\includegraphics{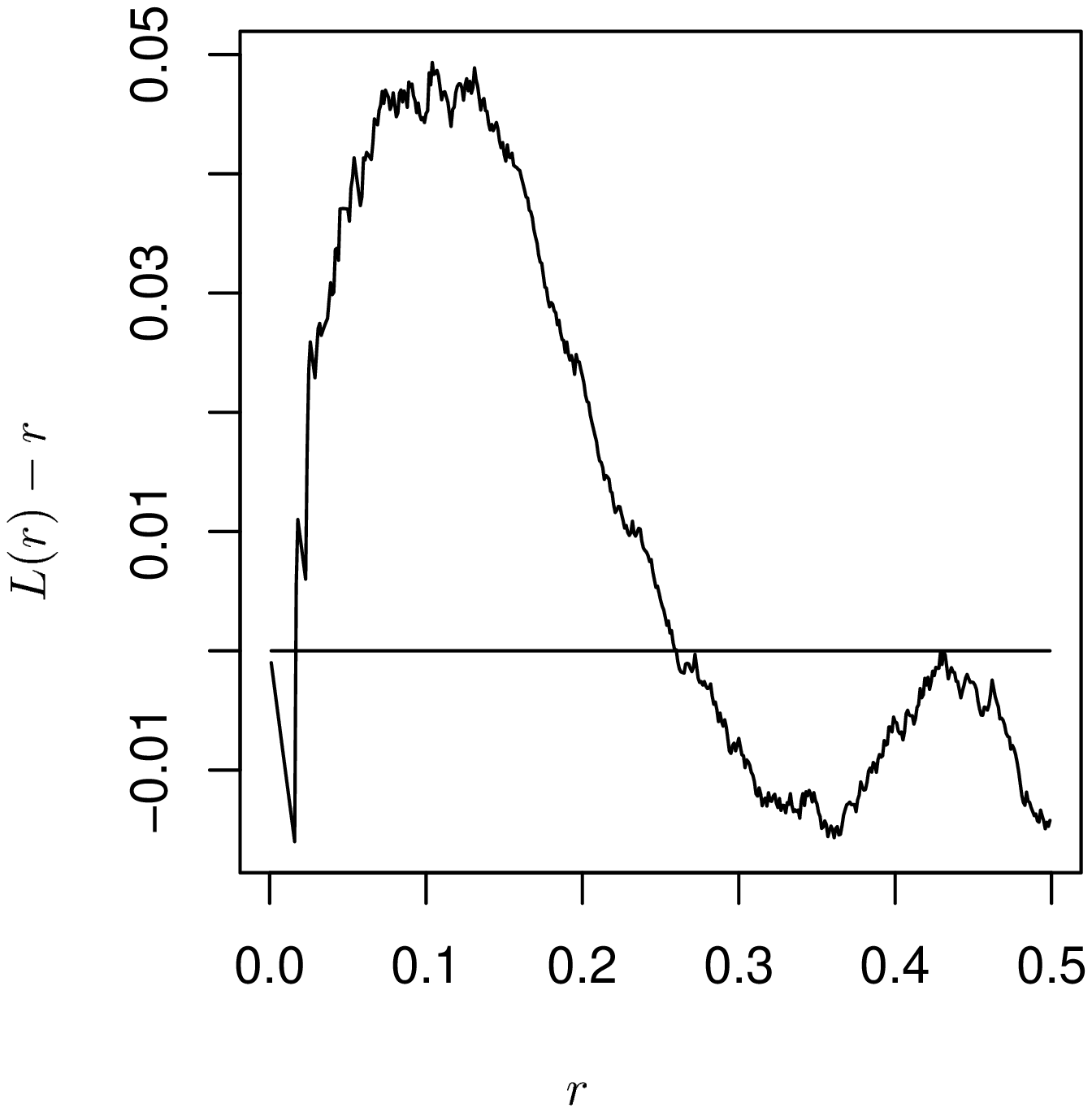}}}
 \end{center}
 \caption[Redwood seedlings data]{Redwood seedlings data.  Left: The data, selected by \cite{rip:mod-spa} from a larger data set analysed by \cite{str:a-mod}.  Right: Plot of the point-process L-function for the redwood seedlings.  There seems to be interaction at 3 different scales: (very) small scale repulsion followed by attraction at a moderate scale and then repulsion at larger scales.}
 \label{fig:red-points}
%\vspace{1ex}
\end{figure}

We take a brief look at a data set which has been much analysed in the 
literature, the Redwood seedlings data first considered by 
\cite{str:a-mod}\index{Strauss, D. J.}.  We examine a subset of the original data chosen by 
\cite{rip:mod-spa}\index{Ripley, B. D.|(} and later analysed
by \cite{dig:par-est}\index{Diggle, P. J.} among 
others.  The data are plotted in Figure \ref{fig:red-points}.  We wish to 
model these data using the multiscale model we have introduced.  
The right pane of Figure~\ref{fig:red-points} gives the estimated point
process L-function\index{point process!L-function|(}\footnote{There is no
connection between the point process L-function and the use of the notation
$L$ elsewhere in this paper for the lower process in the CFTP
algorithm\index{coupling!dominated CFTP algorithm|pagenote}; the clash of notation is an unfortunate result of the standard use of $L$ in both contexts. Nor does either use of $L$ refer to a likelihood.} of the data, defined by $L(t)=\sqrt{\pi^{-1} K(t)}$ where $K$ is the K-function\index{point process!K-function|(} as defined by \citet{rip:sec-ord,rip:mod-spa}\index{Ripley, B. D.|)}.   

From this plot we estimate values of $R_1$ and $R_2$ as $0.07$ and 
$0.013$ respectively, giving repulsion at small scales and attraction at 
moderate scales.  It also seems that there is some repulsion at slightly 
larger scales, so it may be possible to use $R_2=0.2$ and to model the 
large-scale interaction rather than the small-scale interaction as we 
have chosen.

\begin{figure}
 \begin{center}
  \resizebox{.95\textwidth}{!}{\includegraphics{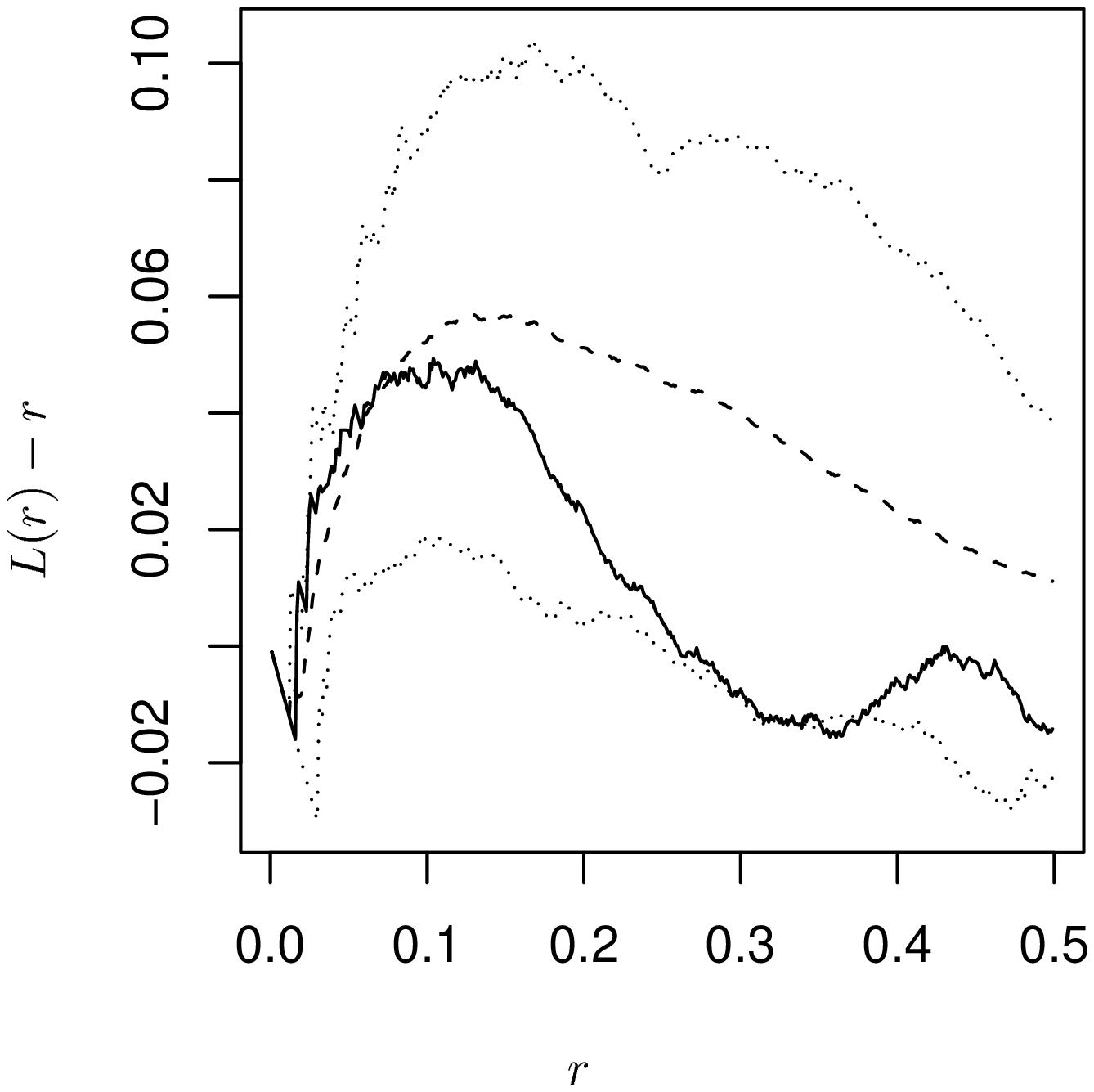}\hspace{2em}\includegraphics{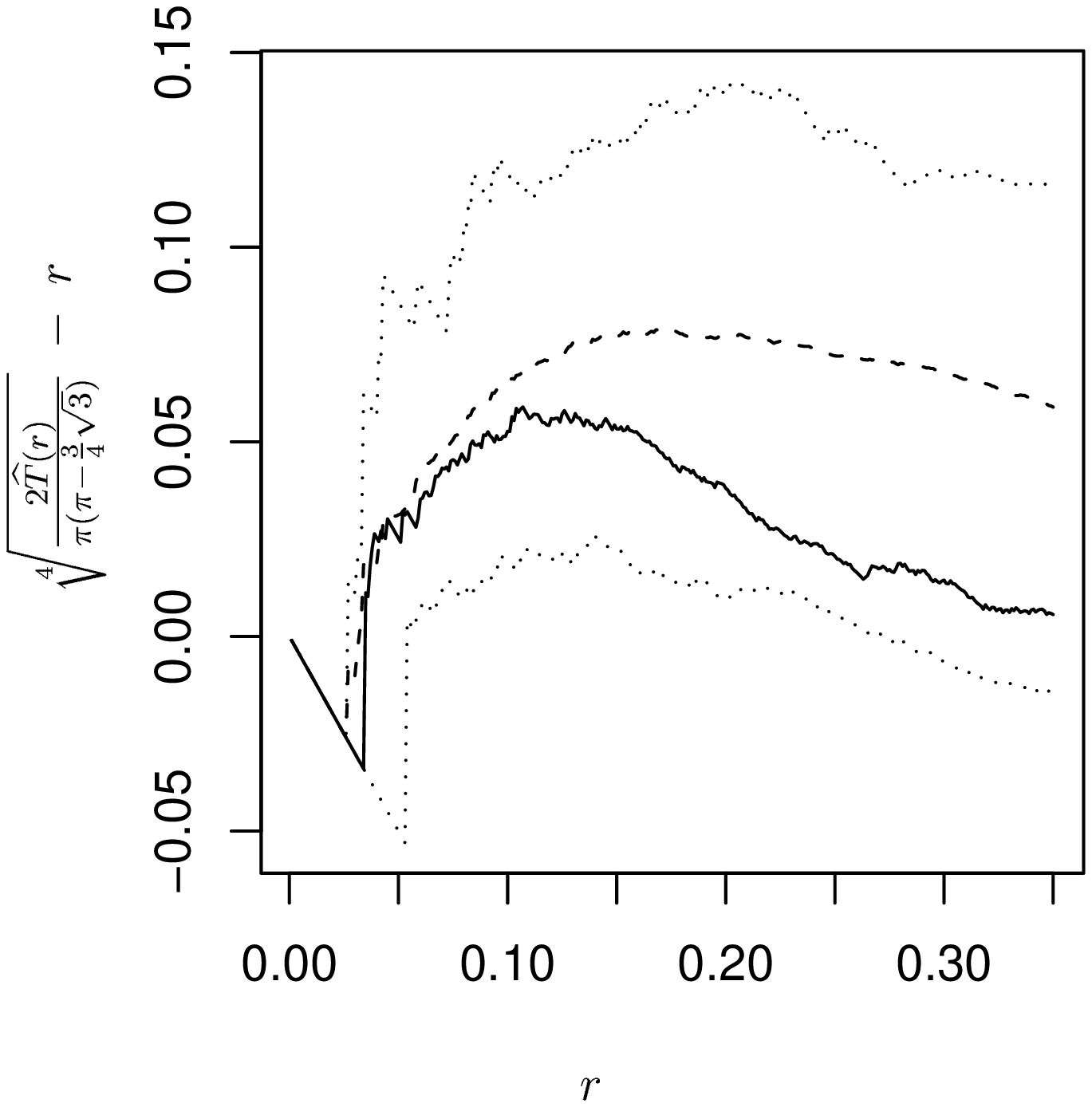}}
 \end{center}
 \caption[L- and T-function plots of the redwood seedlings data]{Point process L- and transformed T-function plots of the redwood seedlings data.  Left: L-function plots of the data together with simulations of the multiscale model with parameters $R_1=0.07$, $R_2=0.013$, $\lambda=0.118$, $\gamma_1=2000$ and $\gamma_2=10^{-200}$.  Dotted lines give an envelope of 19 simulations of the model, the solid line is the redwood seedlings data and the dashed line is the average of the 19 simulations.  Right: the corresponding plots for the transformed T-function.}
 \label{fig:simredk+t}
%\vspace{1ex}
\end{figure}

Experimenting with various values for the remaining parameters, we chose values 
$\gamma_1=2000$ and $\gamma_2=10^{-200}$.  The value $\lambda=0.118$ was chosen to give about 62 points in each realisation, the number in the observed data set.
The remarkably 
small value of $\gamma_2$ was necessary because the value of $R_2$ was 
also very small.  It is clear from these numbers that it would be more 
natural to define $\gamma_1$ and $\gamma_2$ on a logarithmic scale.  
Figure \ref{fig:simredk+t} shows point process L- and T-function plots for 19 
simulations from this model, providing approximate 95\% Monte-Carlo 
confidence envelopes for the values of the functions.  It can be seen that 
on the basis of these functions, the model appears to fit the data 
reasonably well.  The T-function\index{point process!T-function|(}, defined
by \citet{sch-bad:thi-ord}\index{Schladitz, K.}\index{Baddeley, A. J.}, is a third order analogue of the
K-function\index{point process!K-function|)}, and for a Poisson
process\index{Poisson, S. D.!Poisson process} $T(r)$ is proportional to $r^4$; in Figure \ref{fig:simredk+t} the function is transformed by taking the fourth root of a suitable multiple and then subtracting $r$, in order to yield a function whose theoretical value for a Poisson process would be zero.

The plots show several things:  firstly that the model fits reasonably well, 
but that it is possible that we chose a value of $R_1$ which was slightly too 
large.  Perhaps $R_1=0.06$ would have been better.  Secondly, it seems 
that the large-scale repulsion may be an important factor which should not be 
ignored.  Thirdly, in this case we have gained little new information by 
plotting the T-function\index{point process!T-function|)}---the third-order behaviour of the data seems 
to be similar in nature to the second-order structure.

\subsection{Further comments}\label{discuss}

The main advantage of our method for the perfect simulation of locally stable point processes is that it allows acceptance 
probabilities to be computed in $O(n)$ instead of $O(2^{n})$
steps\index{computational complexity} for models 
which are neither purely attractive nor purely repulsive.  Because of the 
exponential dependence on $n$, the algorithm of \cite{ken-mol:per} is not 
feasible in these situations.

It is clear that in practice it is possible to extend the work to more general multiscale 
models.  For example, the sample $L$-function\index{point process!L-function|)} of the 
redwood seedlings might, if the sample size were larger, indicate the 
appropriateness of a three-scale model
\begin{equation}\label{threescale}
	p(X) = \alpha\lambda^{N(X)}\gamma_1^{-m(X\oplus G_1)}
					\gamma_2^{-m(X\oplus G_2)}
					\gamma_3^{-m(X\oplus G_3)}.
\end{equation}
The proof given in the Appendix of \cite{amb-sil:per} can easily be 
extended to show the existence of this process, and (\ref{threescale}) 
is also amenable to perfect simulation using the method of 
Section~\ref{perfect_alg}.  Because of the small size of the redwood 
seedlings data set a model of this complexity is not warranted, but 
the fitting of such models, and even higher order multiscale models in 
appropriate circumstances, would be an interesting topic for future 
research\index{data!Redwood seedlings|)}.

Another topic is the possibility of fitting parameters by a more systematic approach than the subjective adjustment approach we have used.    \citet{amb-sil:per} set out the possibility of using pseudo-likelihood\index{likelihood!pseudo-likelihood} 
\citep{bes:spa-int,bes:sta-ana,bes:met-sta,jen-mol:pse-exp} to 
estimate the parameters $\lambda$, $\gamma_1$ and $\gamma_2$ for given $R_1$ and $R_2$.   However, this method has yet to be implemented and investigated in practice\index{point process!area-interaction process|)}\index{point process!clustering|)}.

%------------------------------------------

\section{Nonparametric regression by wavelet thresholding}
\label{sec:wavelets}
\subsection{Introductory remarks}

\index{regression!nonparametric regression|(}\index{wavelet!thresholding|(}We now turn to our next theme, nonparametric regression.   
Suppose we observe 
\begin{equation}\label{regr}
	y_i = g(t_i)+\varepsilon_i.
\end{equation}
where $g$ is an unknown function sampled with error at regularly
spaced intervals $t_i$.  The noise, $\varepsilon_i$ is assumed to be 
independent and Normally distributed with zero mean and variance $\sigma^2$.

The standard wavelet-based approach to this problem is based on two 
properties of the wavelet transform:
\begin{enumerate}[1.]
\item A large class of `well-behaved' functions can be sparsely represented 
in wavelet space;
\item The wavelet transform maps independent identically distributed
noise to independent identically distributed wavelet coefficients.
\end{enumerate}

These two properties combine to suggest that a good way to remove noise
from a signal is to transform the signal into wavelet space, discard all 
of the small coefficients (i.e. threshold), and perform the inverse 
transform.  Since the true (noiseless) signal had a sparse representation 
in wavelet space, the signal will essentially be concentrated in a small 
number of large coefficients.  The noise, on the other hand, will still be
spread evenly among the coefficients, so by discarding the small 
coefficients we must have discarded mostly noise and will thus have found a 
better estimate of the true signal.

The problem then arises of how to choose the threshold value.  General 
methods that have been applied in the wavelet context are
SureShrink\index{wavelet!SureShrink method@\emph{SureShrink} method} 
\citep{don-joh:ada-unk},
cross-validation\index{wavelet!cross-validation method}
\citep{nas:wav-shr} and 
false discovery rates\index{wavelet!false discovery rate method}
\citep{abr-ben:ada-thr}.  In the
BayesThresh\index{wavelet!BayesThresh method@\emph{BayesThresh} method}
approach \citep{abr-sap-sil:wav-thr} proposes a
Bayesian\index{Bayes, T.!Bayesian hierarchical model} hierarchical
model for the wavelet coefficients, using a mixture of a point mass at $0$ 
and a $N(0,\tau^2)$ density as their prior.  The marginal posterior
\Index{median} 
of the population wavelet coefficient is then used as the estimate.  This 
gives a thresholding rule, since the point mass at $0$ in the prior gives 
non-zero probability that the population wavelet coefficient will be zero.

Most Bayesian approaches to wavelet thresholding model the coefficients 
independently.  In order to capture the notion that nonzero wavelet coefficients may 
be in some way clustered, we allow prior dependency between the coefficients by
modelling them using an extension of the area-interaction process as defined in 
Section \ref{std_aip} above.
The basic idea is that if a coefficient is nonzero then 
it is more likely that its neighbours (in a suitable sense) are also non-zero.
We then use an appropriate CFTP\index{coupling!coupling from the past (CFTP)} approach to  
sample from the posterior distribution of our model.

%An outline of the paper is as follows.  In Section~\ref{model} we briefly survey 
%the area-interaction process and introduce our model for the wavelet coefficients.  
%In Section~\ref{simulation} we discuss coupling from the past, and an extension 
%which allows us to sample from the posterior distribution of our model.  In 
%Section~\ref{simstud} we present a simulation study to compare our method with the
%others introduced in this section.  Section~\ref{concl} presents some conclusions 
%and discusses possible avenues for future work.

\subsection{A Bayesian model for wavelet thresholding}\label{model}\label{ext}

%\subsection{Model specification}

%We describe a novel thresholding procedure which uses a discrete
%area-interaction process to model the correlation between neighbouring
%coefficients in the wavelet transform.

\cite{abr-sap-sil:wav-thr} consider the problem where the true
wavelet coefficients are observed subject to Gaussian noise with zero mean and some
variance $\sigma^2$,
\[
	\widehat{d}_{jk}|d_{jk} \sim N(d_{jk},\sigma^2),
\]
where $\widehat{d}_{jk}$ is the value of the noisy wavelet coefficient
(the data) and $d_{jk}$ is the value of the true (noiseless) coefficient.

Their prior distribution on
the true wavelet coefficients is a mixture of a Normal distribution with
zero mean and variance dependent on the level of the
coefficient, and a point mass at zero as follows:
\begin{equation}\label{BayesThresh}
	d_{jk}\sim\pi_{j} N(0,\tau_j^2)+(1-\pi_{j})\delta(0),
\end{equation}
where $d_{jk}$ is the value of the $k$th coefficient at level $j$ of the
discrete wavelet transform, and the mixture weights $\{\pi_j\}$ are
constant within each level.  An alternative formulation of this can be
obtained by introducing auxiliary variables $Z=\{\zeta_{jk}\}$ with
$\zeta_{jk}\in\{0,1\}$ and independent hyperpriors\index{hyperprior}\undex{Bernoulli, J.}
\begin{equation}\label{indep}
	\zeta_{jk}\sim \text{Bernoulli}(\pi_j).
\end{equation}
The prior given in equation (\ref{BayesThresh}) is then expressed as
\begin{equation}\label{SpaBayesThresh}
	d_{jk}|Z\sim N(0,\zeta_{jk}\tau_j^2).
\end{equation}

The starting point for our extension of this approach is to note that $Z$ can be considered
to be a point process on the discrete space, or lattice, $\chi$ of indices $(j,k)$ of the wavelet coefficients.
The points of $Z$ give the locations at which the prior variance of the wavelet coefficient, conditional on $Z$,
is nonzero.
From this point of view, the hyperprior structure given in equation~(\ref{indep}) is 
equivalent to specifying $Z$ to be a Binomial process with rate function $p(j,k)=\pi_j$.  

Our general approach will be to replace $Z$ by a more general lattice process 
$\xi$ on $\chi$.
We allow $\xi$ to have multiple points at particular locations $(j,k)$, so that the number $\xi_{jk}$ 
of points at $(j,k)$ will be a non-negative integer, not necessarily confined to $\{ 0, 1 \}$. 
We will assume that the prior variance is 
proportional to the number of points of $\xi$ falling at the corresponding 
lattice location.   So if there are no points, the prior will be 
concentrated at zero and the corresponding observed wavelet will be 
treated as pure noise;  on the other hand, the larger the number of 
points, the larger the prior variance and the less shrinkage applied to the 
observed coefficient.
To allow for this generalisation, we extend (\ref{SpaBayesThresh}) in the natural way to
\begin{equation}\label{truecoef}
	d_{jk}| \xi \sim N(0,\tau^2 \xi_{jk}),
\end{equation}
where $\tau^2$ is a constant.  

We now consider the specification of the process $\xi$.  
While it is reasonable that the wavelet transform will produce a sparse
representation, the time-frequency localisation properties of the transform
also make it natural to expect that the representation will be
clustered\index{point process!clustering} in
some sense.  The existence of this clustered structure can be seen clearly in
Figure~\ref{fig:dwt}, which shows the discrete wavelet transform of several
common test functions represented in the natural binary tree configuration.
\begin{figure}
	\begin{center}
		\resizebox{0.95\textwidth}{!}{\includegraphics{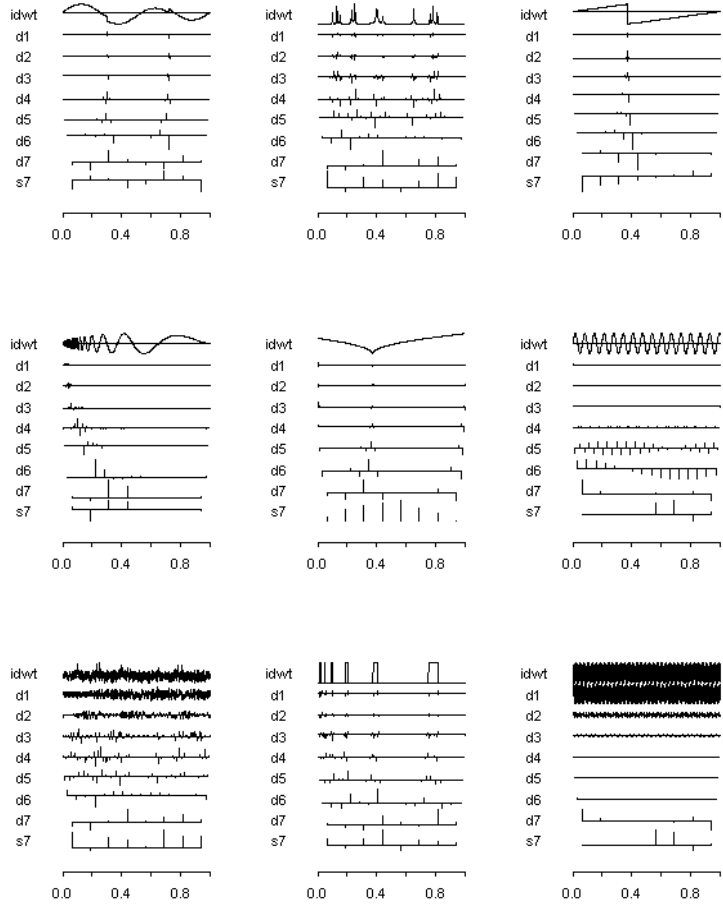}}
	\end{center}
\caption[Examples of the DWT of some test functions.]{Examples of the discrete wavelet transform of some test functions.  
There is clear evidence of clustering in most of the graphs.  The original functions are shown above their discrete wavelet transform each time.\label{fig:dwt}}
\end{figure}
With this clustering in mind, we model $\xi$ as an area-interaction
process\index{point process!area-interaction process} on the space $\chi$.  
The choice of the 
neighbourhoods $B(x)$ for $x$ in $\chi$ will be discussed below.   Given the choice of neighbourhoods, the process will 
be defined by 
\begin{equation}\label{aipprior}
	p(\xi) = \alpha\lambda^{N(\xi)}
				\gamma^{-m\{U(\xi)\}}
\end{equation}
where $p(\xi)$ is the intensity relative to 
the unit rate independent \index{Poisson, S. D.!auto-Poisson process}auto-Poisson process
\citep{cre:sta}.  If
we take $\gamma>1$ this gives a clustered configuration.  Thus we would expect
to see clusters\index{point process!clustering} of large values of $d_{jk}$ if this were a reasonable
model---which is exactly what we do see in Figure~\ref{fig:dwt}.

A simple application of Bayes's theorem\index{Bayes, T.!Bayess theorem@Bayes's theorem} tells us that the posterior for our
model is
\begin{eqnarray}\nonumber
	\lefteqn{p(\xi ,\mathbf{d}|\widehat{\mathbf{d}})  = 
		p(\xi)
		\prod_{j,k}p(d_{jk}|\xi_{jk})
		\prod_{j,k}p(\widehat{d}_{jk}|d_{jk},\xi_{jk})}\\ \nonumber
	& = &	\alpha\lambda^{N(\xi)}\gamma^{-m\{U(\xi)\}}
		\prod_{j,k}\frac{\exp(-d_{jk}^2/2\tau^2 \xi_{jk})}
			{\sqrt{2\pi\tau^2 \xi_{jk}}}
	\prod_{j,k} \frac{\exp\{-(\widehat{d}_{jk}-d_{jk})^2/2\sigma^2 \}}
			{\sqrt{2\pi\sigma^2}}  
			\\ & =& 
			\alpha\lambda^{N(\xi)}\gamma^{-m\{U(\xi)\}}
		\prod_{j,k}\frac{\exp\{ -d_{jk}^2/2\tau^2 \xi_{jk}-(\widehat{d}_{jk}-d_{jk})^2/2\sigma^2 \}}
			{\sqrt{2\pi\tau^2 \xi_{jk}} \sqrt{2\pi\sigma^2}}. 
			\label{post}
\end{eqnarray}

Clearly (\ref{post}) is not a standard density.  In Section \ref{alg} we 
show how the extension of the coupling-from-the-past algorithm described in 
Section~\ref{perfect_alg} enables us to sample from it.  

\subsection{Completing the specification}

We first note that in this context
$\chi$ is a discrete space, so the technical conditions required in Section \ref{std_aip}
of 
$m(\cdot)$ and $B(\cdot)$ are trivially satisfied.

In order to complete the specification of our area-interaction
prior\index{point process!area-interaction process} for $\xi$, 
we need a suitable interpretation of the neighbourhood of a location 
$x= (j,k)$ on the lattice $\chi$ of indices $(j,k)$ of wavelet coefficients.
This lattice is a binary tree, and there are 
many possibilities.  We decided to use the parent, the
coefficient on the parent's level of the transform which is next-nearest
to $x$, the two adjacent coefficients on the level of $x$, the two
children and the coefficients adjacent to them, making a total of nine
coefficients (including $x$ itself).  Figure \ref{fig:tree} illustrates
this scheme, which captures the localisation of both time and frequency
effects.  Figure \ref{fig:tree} also shows
how we dealt with boundaries: we assume that the signal we are
examining is periodic, making it natural to have periodic boundary
conditions in time.
\begin{figure}
	\begin{center}
		\resizebox{0.95\textwidth}{!}{\includegraphics{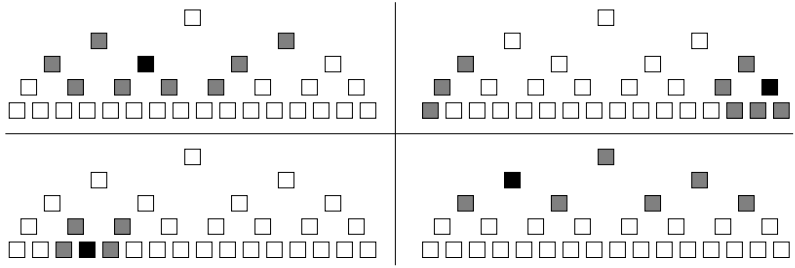}}
	\end{center}
	\caption[Examples of $B(\cdot)$.]{The four plots give examples of what we used as $B(\cdot)$ for four different example locations showing how we dealt with boundaries.  Grey boxes are $B(x)\setminus\{x\}$ for each example location $x$, while $x$ itself is shown as black.\label{fig:tree}}
\end{figure}
If $B(x)$ overlaps with a frequency boundary we simply discard those parts
which have no locations associated with them.  The simple counting measure
used has $m\{B(x)\}=9$ unless $x$ is in the bottom row or one of the top two
rows.

Other possible neighbourhood functions include using only the parent, children
and immediate sibling and cousin of a coefficient as $B(x)$, or a variation on
this taking into account the length of support of the wavelet used.  Though
we have chosen to use periodic boundary conditions, our method is equally
applicable without this assumption, with appropriate modification of $B(x)$.

\section{Perfect simulation for wavelet curve estimation} \label{sec:wavperfsim}
\subsection{Exact posterior sampling for lattice
processes}\label{alg}\index{wavelet!curve estimation}

In this section, we develop a practical approach to simulation from a close approximation to the
posterior density (\ref{post}), making use of coupling from the
past\index{coupling!coupling from the past (CFTP)}.
One of the advantages of the Normal model we propose in Section~\ref{ext} 
is that it is possible to integrate out $d_{jk}$ and work only with the lattice process
$\xi$.  Performing this calculation, we see that equation~(\ref{post}) 
can be rewritten as
\begin{equation*}
	p(\xi|\widehat{\mathbf{d}}) 
%	& = & p(\xi)
%		\prod_{j,k}\int\frac{\exp(-d_{jk}^2/2\tau^2\xi_{jk})}
%			{\sqrt{2\pi\tau^2\xi_{jk}}}
%			\frac{\exp\left\{-(\widehat{d}_{jk}-d_{jk})^2/2\sigma^2\right\}}
%			{\sqrt{2\pi\sigma^2}}dd_{jk}\\
% 	&\!\!\!=\!\!\!&	p(\xi)\prod_{j,k}\int
% 			\frac{\exp(-d_{jk}^2/2\tau^2\xi_{jk}-
% 				(\widehat{d}_{jk}-d_{jk})^2/2\sigma^2)}
% 			{\sqrt{4\pi^2\tau^2\xi_{jk}\sigma^2}}dd_{jk}\\
%	&\!\!\!=\!\!\!&	p(\xi)\prod_{j,k}\int
%			\frac{\exp\left[-\left\{d_{jk}^2(\sigma^2+\tau^2\xi_{jk})
%				-2d_{jk}\widehat{d_{jk}}\tau^2\xi_{jk}
%				+\widehat{d_{jk}}\tau^2\xi_{jk}\right\}/2\tau^2\xi_{jk}\sigma^2\right]}
%			{\sqrt{4\pi^2\tau^2\xi_{jk}\sigma^2}}dd_{jk}\\
% 	&\!\!\!=\!\!\!&	p(\xi)\prod_{j,k}\int\frac{\exp\left(
% 			-\left(\frac{\sigma^2+\tau^2\xi_{jk}}{2\tau^2\xi_{jk}\sigma^2}\right)
% 			\left(d_{jk}-\frac{\widehat{d_{jk}}\tau^2\xi_{jk}}{\sigma^2+\tau^2\xi_{jk}}
% 				\right)^2 - \frac{-\widehat{d_{jk}}^2}{2(\sigma^2+\tau^2\xi_{jk})}
% 			\right)}
% 			{\sqrt{4\pi^2\tau^2\xi_{jk}\sigma^2}}dd_{jk}\\
%	&\!\!\!=\!\!\!&	p(\xi)\prod_{j,k}
%		\frac{\exp\left\{\frac{-\widehat{d_{jk}}^2}{2(\sigma^2+\tau^2\xi_{jk})}\right\}}
%		{\sqrt{4\pi^2\tau^2\xi_{jk}\sigma^2}}
%			\int\exp\left\{
%			-\left(\frac{\sigma^2+\tau^2\xi_{jk}}{2\tau^2\xi_{jk}\sigma^2}\right)
%			\left(d_{jk}-\frac{\widehat{d_{jk}}\tau^2\xi_{jk}}{\sigma^2+\tau^2\xi_{jk}}
%				\right)^2\right\}dd_{jk}\\
%	&\!\!\!=\!\!\!&	p(\xi)\prod_{j,k}
%		\frac{\exp\left\{\frac{-\widehat{d_{jk}}^2}{2(\sigma^2+\tau^2\xi_{jk})}\right\}}
%		{\sqrt{4\pi^2\tau^2\xi_{jk}\sigma^2}}
%		\left(\frac{2\pi\sigma^2\tau^2\xi_{jk}}{\sigma^2+\tau^2\xi_{jk}}\right)^{1/2}\\
	=	p(\xi)\prod_{j,k}\frac{\exp
			\left\{-\widehat{d}_{jk}^2/2(\sigma^2+\tau^2\xi_{jk})\right\}}
		{\sqrt{2\pi(\sigma^2+\tau^2\xi_{jk})}},
\end{equation*}
by the standard convolution properties of normal densities.
We now see that it is possible to sample from the posterior by simulating only 
the process $\xi$ and ignoring the marks $\mathbf{d}$.  This lattice 
process is amenable to perfect simulation using the method of Section \ref{perfectmulti} above.
Let
\begin{align*}
	f_1(\xi)= &\; \lambda^{N(\xi)},\\
	f_2(\xi)= &\; \gamma^{-m\left\{U(\xi)\right\}},\\
	f_3(\xi)= &\; \prod_{j,k}
		\exp \{-\widehat{d}_{jk}^2\left/\right.2(\sigma^2+\tau^2\xi_{jk})\}
		\text{ and}\\
	f_4(\xi)= &\; \prod_{j,k}\left\{2\pi(\sigma^2+\tau^2\xi_{jk})\right\}^{-1/2}.
\end{align*}
Then
\begin{align*}
	\lambda_{f_1}(u;\xi)= &\; \lambda,\\
	\lambda_{f_2}(u;\xi)= &\; \gamma^{-m\left\{B(u)\setminus U(\xi)\right\}}\leq1,\\
	\lambda_{f_3}(u;\xi)= &\; 
		\exp\left\{\frac{\widehat{d}_u^2\tau^2}
			{2(\sigma^2+\tau^2\xi_{u})\{\sigma^2+\tau^2(\xi_{u}+1)\}}\right\} \\
			& \: \leq\exp\left\{\frac{\widehat{d}_u^2\tau^2}
			{2\sigma^2(\tau^2+\sigma^2)}\right\}
\text{ and}\\
	\lambda_{f_4}(u;\xi)= &\;
		\left\{\frac{\tau^2\xi_u+\sigma^2}{\tau^2(\xi_u+1)+\sigma^2}\right\}^{1/2}
			\leq1.
\end{align*}
By a slight abuse of notation, in the second and third equations above we use 
$u$ to refer both to the point $\{u\}$ and the location $(j,k)$ at which it 
is found.
The functions $f_1$, \ldots, $f_4$ are also monotone with respect to the subset relation, 
so all of the conditions for exact simulation using the method of
Section~\ref{perfect_alg} are satisfied.

In the spatial processes considered in detail in Section \ref{perfectmulti}, the 
dominating process had constant intensity across the space $\chi$. 
In the present context, however, it is necessary in practice to  use a dominating process which 
has a different rate at each lattice location, and then use 
location-specific maxima
and minima rather than global maxima and minima.   Because we can now 
use location-specific, rather than global,
maxima and minima, we can initialise upper and lower processes that are much
closer together than would have been possible with a constant-rate 
dominating process.  This has the consequence of reducing \Index{coalescence} times
to feasible levels.
A constant-rate dominating process would not have been 
feasible due to the size of
the global maxima, so this modification to the method of Section \ref{perfectmulti}
is essential; see Section~\ref{large+small1} for details.  
Chapter 5 of \cite{amb:dom-cou} gives some other 
examples of dominating processes with location-specific intensities.

% The dominating process used in \cite{amb-sil:per} is constant across 
% $\chi$.  Taking advantage of the lattice structure, we use a dominating 
% process which has a different rate at each location.  This enables us 
% to initialise the upper and lower processes much closer together than 
% would have been possible had we used a dominating process whose rate was 
% the same at each location, as we can now use location-specific maxima and 
% minima rather than global maxima and minima.  As a result, the average 
% coalescence time of the algorithm is significantly smaller.  See 
% \cite{amb:dom-cou} Chapter 5 for some other examples of this kind of 
% dominating process.  In fact, the difference in performance is so great 
% that using a constant rate dominating process would not have been feasible 
% due to the size of the global maxima.  See Section~\ref{large+small1} for 
% details.  

The location-specific rate of the dominating process $D$ is 
\begin{equation}\label{lambda_dom}
	\lambda_{jk}^{dom} = \lambda 
		       e^{\widehat{d}_{jk}^2\tau^2/2\sigma^2(\tau^2+\sigma^2)}
\end{equation}
for each location $(j,k)$ on the lattice.  The lower process is then started 
as a thinned version of $D$.  Points are accepted with probability
\[
	P(x) = \gamma^{-M(\chi)}\left(\frac{\sigma^2}{\tau^2+\sigma^2}\right)^{1/2}
		\times\exp\left\{ - 
			\frac{\widehat{d}_x^2\tau^2}
				{2\sigma^2(\tau^2+\sigma^2)}\right\},
\]
where $M(\chi) = \max_\chi[m\{B(x)\}]$.
The upper and lower processes are then evolved through time, 
accepting points as described in Section~\ref{perfect_alg} with probability
\[
\frac1{\lambda_{jk}^{dom}}\lambda_{f_1}(u;\xi^{\text{up}})
		\lambda_{f_2}(u;\xi^{\text{up}})
		\lambda_{f_3}(u;\xi^{\text{low}})
		\lambda_{f_4}(u;\xi^{\text{up}})
\]
for the upper process and
\[
\frac1{\lambda_{jk}^{dom}}\lambda_{f_1}(u;\xi^{\text{low}})
		\lambda_{f_2}(u;\xi^{\text{low}})
		\lambda_{f_3}(u;\xi^{\text{up}})
		\lambda_{f_4}(u;\xi^{\text{low}})
\]
for the lower process.  The remainder of the algorithm carries
over in the obvious way.  
There are still some issues to be addressed due to very high birth rates 
in the dominating process, and this will be done in Section \ref{large+small1}.

\subsection{Using the generated samples}\label{using}

Although $\mathbf{d}$ was integrated out for simulation reasons in
Section~\ref{ext} it is, naturally, the quantity of interest.  Having
simulated realisations of $\xi|\widehat{\mathbf{d}}$ we then
generate $\mathbf{d}|\xi,\widehat{\mathbf{d}}$ for each
realisation $\xi$ generated in the first step.  The sample
\Index{median} of $\mathbf{d}|\xi,\widehat{\mathbf{d}}$ gives an estimate for
$\mathbf{d}$.  The median is used instead of the mean as this
gives a thresholding rule, defined by \cite{abr-sap-sil:wav-thr} as a rule giving $p(d_{jk}=0|\widehat{\mathbf{d}})>0$.

We calculate $p(\mathbf{d}|\xi,\widehat{\mathbf{d}})$ using logarithms
for ease of notation.  Assuming that $\xi_{jk}\ne 0$ we find
\begin{eqnarray*}
	\log p(d_{jk}|\widehat{d}_{jk},\xi_{jk}\ne 0) & = &
		\log{p(d_{jk}|\xi_{jk}\ne0)} +
		\log{p(\widehat{d}_{jk}|d_{jk},\xi_{jk}\ne0)} + C\\
	& = & \frac{-d_{jk}^2}{2\tau^2\xi_{jk}}+
			\frac{-(\widehat{d}_{jk}-d_{jk})^2}{2\sigma^2} + C_1\\
	& = & -\frac{(\sigma^2+\tau^2\xi_{jk})
		\left(d_{jk}-\frac{\tau^2\xi_{jk} \widehat{d}_{jk}}
			{\sigma^2+\tau^2\xi_{jk}}\right)^2}
		{2\sigma^2\tau^2\xi_{jk}}+C_2
\end{eqnarray*}
where $C$, $C_1$ and $C_2$ are constants.  Thus
\[
	d_{jk}|\widehat{d}_{jk},\xi_{jk}\ne 0 \sim
		N\left(\frac{\tau^2\xi_{jk} \widehat{d}_{jk}}
			{\sigma^2+\tau^2\xi_{jk}},
		\frac{\sigma^2\tau^2\xi_{jk}}{\sigma^2+\tau^2\xi_{jk}}\right).
\]
When $\xi_{jk}=0$ we clearly have $p(d_{jk}|\xi_{jk},\widehat{d}_{jk})=0$.

%Although, then, we do not have the same theoretical justification (a
%further reason is that in our model we have assumed that the wavelet
%coefficients are \emph{not} independent of one another), we do make use
%of the posterior median, since it gives a thresholding rule, thus taking
%advantage of the natural sparsity of the wavelet transform.

\subsection{Dealing with large and small rates}\label{large+small1}

We now deal with some approximations which are necessary to allow our
algorithm to be feasible computationally\index{computational complexity}.
Recall from equation (\ref{lambda_dom})
that if the maximum data value $d_{jk}$ is twenty times larger in magnitude
than the standard deviation of the noise (a not uncommon event for
reasonable noise levels) then we have
\begin{eqnarray*}
	\lambda_{dom} & = & \lambda e^{400\sigma^2\tau^2/
						2\sigma^2(\tau^2+\sigma^2)}\\
		& = & \lambda e^{200\tau^2/(\tau^2+\sigma^2)}.
\end{eqnarray*}
Now unless $\tau$ is significantly smaller than $\sigma$, this will result in
enormous birth rates, which make it necessary to modify the algorithm appropriately.  
To address this issue, we noted that 
the chances of there being no
live points at a location whose data value is large (resulting in a value
of $\lambda_{dom}$ larger than $e^4$) is sufficiently small that for the
purposes of calculating $\lambda_{f_2}(u;\xi)$ for
nearby locations it can be assumed that the number of points alive was
strictly positive.  

This means that we do not know the true value
%Finally, 
%recall from Section~\ref{large+small1} that we may not know the value
of $\xi_{jk}$ for the locations with the largest values of $d_{jk}$.  This
leads to problems since we need to generate $d_{jk}$ from the distribution
\[
	d_{jk}|\xi_{jk},\widehat{d}_{jk}
		\sim
			N\left(\frac{\tau^2\xi_{jk}\widehat{d}_{jk}}
				{\sigma^2+\tau^2\xi_{jk}},
			\frac{\sigma^2\tau^2\xi_{jk}}
				{\sigma^2+\tau^2\xi_{jk}}\right),
\]
which requires values of $\xi_{jk}$ for each location $(j,k)$ in the
configuration.
To deal with this issue, we first note that, as $\xi_{jk}\to\infty$,
\[
	\frac{\tau^2\xi_{jk}\widehat{d}_{jk}}{\sigma^2+\tau^2\xi_{jk}}
		\hspace{.5em} \longrightarrow \hspace{.5em}
			\widehat{d}_{jk} 
\]
monotonically from below, and 
\[
	\frac{\tau^2\xi_{jk}\sigma^2}{\sigma^2+\tau^2\xi_{jk}}
		\hspace{.5em} \longrightarrow \hspace{.5em}
			\sigma^2,
\]
also monotonically from below.  Since $\sigma$ is typically small,
convergence is very fast indeed.  Taking $\tau=\sigma$ as an example we see
that even when $\xi_{jk}=5$ we have
\[
	\frac{\tau^2\xi_{jk}\widehat{d}_{jk}}{\sigma^2+\tau^2\xi_{jk}}
		=
			\frac56\widehat{d}_{jk}
\]
and
\[
	\frac{\tau^2\xi_{jk}\sigma^2}{\sigma^2+\tau^2\xi_{jk}}
		=
			\frac56\sigma^2.
\]
We see that we are already within $\frac16$
of the limit.  Convergence is even faster for larger values of $\tau$.

We also recall that the dominating process
gives an upper bound for the value of $\xi_{jk}$ at every location.  Thus a
good estimate for $d_{jk}$ would be gained by taking the value of $\xi_{jk}$
in the dominating process for those points where we do not know the exact
value.  This is a good solution but is unnecessary in some cases, as
sometimes the value of $\lambda_{dom}$ is so large that there is little
advantage in using this value.
Thus for exceptionally large values of $\lambda_{dom}$ we simply use
$N(\widehat{d}_{jk},\sigma^2)$ numbers as our estimate of $d_{jk}$.

\subsection{Simulation study}\label{simstud}\index{wavelet!curve estimation|(}

We now present a simulation study of the performance of our
estimator relative to several established wavelet-based estimators.  Similar
to the study of \cite{abr-sap-sil:wav-thr}\index{Abramovitch,
F.}\index{Benjamini, Y.}, we investigate the 
performance of our method on the four standard test
functions\index{data!standard test functions} of 
%\citeANP{don-joh:ide-spa}~
\cite{don-joh:ide-spa,don-joh:ada-unk}\index{Donoho, D. L.}\index{Johnstone, I. M.}, 
namely `Blocks', `Bumps', `Doppler' and `Heavisine'.  These test 
functions are used because they exhibit different kinds of behaviour typical of 
signals arising in a variety of applications.  

The test functions were simulated at 256 points equally spaced on the unit
interval.  The test signals were centred and scaled so as to have mean value
$0$ and standard deviation $1$.  We then added independent $N(0,\sigma^2)$
noise to each of the functions, where $\sigma$ was taken as $1/10$, $1/7$
and $1/3$.  The noise levels then correspond to root signal-to-noise ratios
(RSNR)\index{signal-to-noise ratio!root signal-to-noise ratio (RSNR)} of $10$, $7$ and $3$ respectively.  We performed 25 replications.  For
our method, we simulated 25 independent draws from the posterior distribution
of the $d_{jk}$ and used the sample \Index{median} as our estimate, as this gives a
thresholding rule.  For each of the runs, $\sigma$ was set to the standard
deviation of the noise we added, $\tau$ was set to $1.0$, $\lambda$ was set
to $0.05$ and $\gamma$ was set to $3.0$.

The values of parameters $\sigma$ and $\tau$ were
set to the true values of the standard deviation of the noise and the
signal, respectively.  In practice it will be necessary to develop some
method for estimating these values.  The value of $\lambda$ was chosen to
be $0.05$ because it was felt that not many of the coefficients would be
significant.  The value of $\gamma$ was chosen based on small trials for the
heavisine and jumpsine datasets.

We compare our method with several established wavelet-based estimators
for reconstructing noisy signals:
SureShrink\index{wavelet!SureShrink method@\emph{SureShrink} method}
\citep{don-joh:ide-spa},
two-fold
cross-validation\index{wavelet!cross-validation method} as applied by
\cite{nas:wav-shr}\index{Nason, G. P.}, ordinary
BayesThresh\index{wavelet!BayesThresh method@\emph{BayesThresh} method}
\citep{abr-sap-sil:wav-thr}, and the false
discovery rate\index{wavelet!false discovery rate method} as applied by
\cite{abr-ben:ada-thr}.
% SureShrink is a method which
% minimises Stein's
% unbiased estimate of risk \citep{ste:est} and a different threshold is chosen
% for each level of the transform.  The way cross-validation is used is due to
% \cite{nas:wav-shr}.  The data are split into odd- and even-numbered
% observations (we shall call them the `odds' and the `evens').  First the
% evens are used to get an estimator for the function (using some threshold
% $t$) and the sum of squared errors (SSE) between the estimate and the odds
% is calculated.  Secondly, the odds are used to get an estimator of the
% function using the same threshold $t$, and the SSE between the new estimate
% and the evens is calculated.  Finally, the combined SSE is minimised
% numerically over values of $t$.  The way that the false discovery rate is
% used is due to \cite{abr-ben:ada-thr}.  They use the methodology to control
% the expected number of coefficients which are not thresholded but should
% have been.

For test signals `Bumps', `Doppler' and
`Heavisine' we used Daube\-chies'\index{Daubechies, I.} least asymmetric wavelet of order 10
\citep{dau:ten}.  For the `Blocks' signal we used the Haar
wavelet\index{Haar, A.!Haar wavelet}, as the
original signal was piecewise constant.  The analysis was carried out
using the freely available $R$ statistical package\index{R statistical
package@$R$ statistical package}.  The
WaveThresh\index{wavelet!WaveThresh package@\emph{WaveThresh} package} package
\citep{nas:wav-thr} was used to perform the discrete wavelet transform and
also to compute the SureShrink, cross-validation, BayesThresh and false
discovery rate estimators.

\begin{table}
\begin{minipage}{96mm}
\caption[Comparison of our estimator with other wavelet-based estimators]{\label{tab:sim}Average mean-square errors ($\times 10^4$) for the area-interaction BayesThresh (AIBT), SureShrink (SS), 
cross-validation (CV), ordinary BayesThresh (BT) and false discovery rate (FDR) estimators for four test functions for three values of the root signal-to-noise ratio.  
Averages are based on 25 replicates.  
Standard errors are given in parentheses.}
\end{minipage}
%\centering
 \begin{tabular}{ccr@{ }rr@{ }rr@{ }rr@{ }r@{\hspace{.7em}}c}
 \hline
  RSNR	& Method	& \multicolumn{9}{c}{Test functions}\\ \cline{3-11}
	&	& \multicolumn{2}{c}{Blocks}   & \multicolumn{2}{c}{Bumps}    & \multicolumn{2}{c}{Doppler} & \multicolumn{3}{c}{Heavisine}\\ \hline
 	& AIBT	& 25 & (1)    & 84 & (2)    & 49 & (1)   & \hspace{1em}32 & (1)&\\
  	& SS	& 49 & (2)    & 131 & (6)   & 54 & (2)   & 66 & (2)\\
  10	& CV	& 55 & (2)    & 392 & (21)  & 112 & (5)  & 31 & (1)\\
  	& BT	& 344 & (10)  & 1651 & (17) & 167 & (5)  & 35 & (2)\\
  	& FDR	& 159 & (14)  & 449 & (17)  & 145 & (5)  & 64 & (3)\\[1ex]
  	& AIBT	& 56 & (3)    & 185 & (5)   & 87 & (3)   & 52 & (2)\\
  	& SS	& 98 & (3)    & 253 & (10)  & 99 & (4)   & 94 & (4)\\
  7	& CV	& 96 & (3)    & 441 & (25)  & 135 & (6)  & 54 & (3)\\
  	& BT	& 414 & (11)  & 1716 & (21) & 225 & (6)  & 57 & (2)\\
  	& FDR	& 294 & (18)  & 758 & (27)  & 253 & (9)  & 93 & (4)\\[1ex]
	& AIBT	& 535 & (21)  & 1023 & (15) & 448 & (18) & 153 & (6)\\
	& SS	& 482 & (13)  & 973 & (45)  & 399 & (14) & 147 & (3)\\
  3	& CV	& 452 & (11)  & 914 & (34)  & 375 & (13) & 148 & (6)\\
	& BT	& 860 & (24)  & 2015 & (37) & 448 & (12) & 140 & (4)\\
	& FDR	& 1230 & (52) & 2324 & (88) & 862 & (31) & 148 & (3)\\ \hline
 \end{tabular}
\end{table}

The goodness of fit of each estimator was measured by its average
mean-square error (AMSE)\index{error!average mean-square error (AMSE)} over the 25 replications.  Table \ref{tab:sim}
presents the results.  It is clear that our estimator performs extremely
well with respect to the other estimators when the \Index{signal-to-noise ratio} is
moderate or large, but less well, though still competitively, when there is a
small signal-to-noise ratio\index{wavelet!curve estimation|)}.  

\subsection{Remarks and directions for future work}\label{concl}

Our procedure for Bayesian wavelet thresholding has used
the naturally clustered nature of the wavelet transform when deciding
how much weight to give coefficient values.  
In comparisons with other methods, our approach performed very
well for moderate and low noise levels, and reasonably competitively for
higher noise levels.

One possible area for future work would be to replace equation
\eqref{truecoef} with
\[
	d_{jk}|\xi \sim N(0,\tau^2(\xi_{jk})^z),
\]
where $z$ would be a further parameter.  This would modify the number of
points which are likely to be alive at any given location and thus also
modify the tail behaviour of the prior. %discussed in Section~\ref{tails}.
The idea behind this suggestion is that when we know that the behaviour of
the data is either heavy\index{tail behaviour!heavy-tailed} or light tailed\index{tail behaviour!light-tailed}, we could adjust $z$ to compensate.
This could possibly also help speed up convergence by reducing the number
of points at locations with large values of $d_{jk}$.  

A second possible area for future work would be to develop some automatic
methods for choosing the parameter values, perhaps using the method of
maximum pseudo-likelihood\index{likelihood!pseudo-likelihood} \citep{bes:spa-int,bes:sta-ana,bes:met-sta}. 

Finally, it would be of obvious interest to find an approach which made the
approximations of Section~\ref{large+small1} unnecessary and allowed for true
CFTP\index{coupling!coupling from the past (CFTP)} to be preserved\index{simulation!exact/perfect simulation|)}\index{regression!nonparametric regression|)}\index{wavelet!thresholding|)}.

\section{Conclusion}

This paper, based on \cite{amb-sil:wave, amb-sil:per}, has drawn 
together a number of themes which demonstrate the way that modern
\index{statistics!computational statistics}computational statistics has made use of work in applied \Index{probability} and stochastic processes
in ways which would have been inconceivable not many decades ago.   
It is therefore a particular pleasure to dedicate it to John
Kingman\index{Kingman, J. F. C.} on his birthday!

%\acks
%The first author would like to thank Guy Nason and Paul Northrop 
%for helpful discussions.
\bibliography{festpaper}
\bibliographystyle{cambridgeauthordateCMG}

\end{document}